\documentclass[10pt, a4paper]{article}
\usepackage{graphicx,latexsym, amsmath,amsfonts}
\usepackage{amsthm}
\usepackage{enumerate}
\usepackage{epsfig}

\begin{document}


\newcommand{\E}{\mathbb{E}}
\newcommand{\PP}{\mathbb{P}}
\newcommand{\RR}{\mathbb{R}}
\newcommand{\NN}{\mathbb{N}}
\newcommand{\fracs}[2]{{ \textstyle \frac{#1}{#2} }}

\newtheorem{theorem}{Theorem}[section]
\newtheorem{lemma}[theorem]{Lemma}
\newtheorem{coro}[theorem]{Corollary}
\newtheorem{defn}[theorem]{Definition}
\newtheorem{assp}[theorem]{Assumption}
\newtheorem{expl}[theorem]{Example}
\newtheorem{prop}[theorem]{Proposition}
\newtheorem{rmk}[theorem]{Remark}
\newtheorem{notation}[theorem]{Notation}

\def\a{\alpha} \def\g{\gamma}
\def\e{\varepsilon} \def\z{\zeta} \def\y{\eta} \def\o{\theta}
\def\vo{\vartheta} \def\k{\kappa} \def\l{\lambda} \def\m{\mu} \def\n{\nu}
\def\x{\xi}  \def\r{\rho} \def\s{\sigma}
\def\p{\phi} \def\f{\varphi}   \def\w{\omega}
\def\q{\surd} \def\i{\bot} \def\h{\forall} \def\j{\emptyset}

\def\be{\beta} \def\de{\delta} \def\up{\upsilon} \def\eq{\equiv}
\def\ve{\vee} \def\we{\wedge}

\def\t{\tau}

\def\F{{\cal F}}
\def\T{\tau} \def\G{\Gamma}  \def\D{\Delta} \def\O{\Theta} \def\L{\Lambda}
\def\X{\Xi} \def\S{\Sigma} \def\W{\Omega}
\def\M{\partial} \def\N{\nabla} \def\Ex{\exists} \def\K{\times}
\def\V{\bigvee} \def\U{\bigwedge}

\def\1{\oslash} \def\2{\oplus} \def\3{\otimes} \def\4{\ominus}
\def\5{\circ} \def\6{\odot} \def\7{\backslash} \def\8{\infty}
\def\9{\bigcap} \def\0{\bigcup} \def\+{\pm} \def\-{\mp}
\def\<{\langle} \def\>{\rangle}

\def\lev{\left\vert} \def\rev{\right\vert}
\def\1{\mathbf{1}}

\newcommand\wD{\widehat{\D}}

\title
{ \bf First order strong approximations  of  scalar  SDEs with values in a domain}

\author{Andreas Neuenkirch\footnote{Institut f\"ur Mathematik, Universit\"at Mannheim,
A5, 6, D-68131 Mannheim, Germany, \texttt{neuenkirch@kiwi.math.uni-mannheim.de}}
 \and Lukasz Szpruch\footnote{Mathematical Institute, University of Oxford, 24-29 St Giles',
Oxford OX1 3LB, England \texttt{Szpruch@maths.ox.ac.uk}} }

\date{}

\maketitle

\begin{abstract}
{\textsf{\em We are interested in strong approximations of one-dimensional SDEs which have  non-Lipschitz coefficients and which take values in a domain. Under a set of general assumptions we derive an implicit scheme  that preserves the domain of the SDEs and is strongly convergent with rate one. Moreover, we show that this general result can be applied to many SDEs we encounter in mathematical finance and  bio-mathematics. We will demonstrate flexibility of our approach by analysing classical examples of SDEs with  sublinear coefficients (CIR, CEV models and Wright-Fisher diffusion) and also with superlinear coefficients (3/2-volatility, Ait-Sahalia model).       
  Our goal is to justify an efficient 
Multi-Level Monte Carlo (MLMC) method for a rich family of SDEs, which relies on good strong convergence properties. 
}

\medskip
\noindent \textsf{{\bf Key words: } \em 
Stochastic differential equations, non-Lipschitz coefficients, Lamperti Transformation,
backward Euler-Maruyama scheme.}

\medskip
\noindent{\small\bf 2000 Mathematics Subject Classification: } 60H10,\;65J15
}
\end{abstract}

\section{Introduction}
The goal of this paper is to derive  an efficient numerical approximation for  one-dimensional SDEs which take values in
a domain and have  non-Lipschitz drift or  diffusion coefficients. 
 Typical examples of such SDEs  are the Cox-Ingersoll-Ross process (CIR), the CEV model, 
the Wright-Fisher diffusion, where the main difficulty is the sublinearity of the diffusion coefficient. 
Furthermore, the approach developed in this paper can be also applied to SDEs with  superlinear coefficients.
Prominent examples are here
the Heston 3/2-volatility process and the Ait-Shalia model. All the mentioned processes play an important role in 
mathematical finance and bio-mathematical applications. 
Our key idea is 
to transform the original SDE using the Lamperti transformation into a SDE with constant diffusion coefficient, see e.g.  \cite{iacus2008simulation}. 
The transformed SDE is then approximated by a   backward (also called drift-implicit) Euler-Maruyama  scheme (BEM)
and transforming back yields an approximation scheme for the original SDE.
 This strategy was found successful in a recent work \cite{dereich2012euler}
for the CIR process,
where the authors proved that the piecewise linearly interpolated BEM scheme strongly converges 
with  rate one  half (up to a log-term) with respect to a uniform $L^p$-error criterion. This strategy was also 
suggested by Alfonsi in \cite{alfonsi2005discretization}. Here, we extend that work in several ways: 
\begin{itemize}
\item Considering the maximum error in the discretization points, we prove that  the drift-implicit Euler-Maruyama scheme for the CIR process strongly converges  with rate one under slightly more restrictive
conditions on the parameters of the process than in \cite{dereich2012euler}.
\item We provide a general framework for  the strong order one convergence of the BEM scheme for  SDEs with constant diffusion 
and one-sided Lipschitz drift coefficients. 
\item Using this framework we present a detailed convergence analysis for several SDEs with  sub- and super-linear coefficients.    
\item We also show that BEM for the transformed SDE is closely related to a drift-implicit Milstein scheme
for the original SDE, which has been introduced in \cite{higham2012convergence}. In the case of the CIR process we provide a sharp error analysis for this scheme.
\end{itemize}

 Independently of  and simultaneously to the research presented in this paper, the same approach was also used by Alfonsi in \cite{alfonsi2012}
to derive strong order one convergence of the BEM scheme for the CIR and the CEV process. See Remark \ref{comp_alf_1} for a discussion.

\medskip

To illustrate the main difficulties and also our main idea let us consider the 
CIR process 
\begin{equation} \label{eq:orCIR}
dy(t) = \k(\o-y(t))dt + \s\sqrt{y(t)}d w(t)
\end{equation}
with $2\k \o \ge\s^2$.
It is a simple implication of the Feller test that the solution of 
equation \eqref{eq:orCIR} is strictly positive when $2\k \o \ge\s^2$ and $y(0)>0$.
This SDE is often used in mathematical finance for interest rate or stochastic volatility models.
However scalar SDEs  with  square root diffusion coefficients
appear not only in the financial literature but belong to the most fundamental SDEs as they are an
approximation to Markov jump processes \cite{either1986markov}.

\smallskip

Once we attempt to simulate \eqref{eq:orCIR}  using classical discretization
methods, see e.g. \cite{kloeden1992numerical}, we face two difficulties:
\begin{itemize}
\item In general, these methods do not preserve positivity and therefore
are not well defined when directly applied to equation \eqref{eq:orCIR};
\item The diffusion term is not Lipschitz continuous and therefore
standard
assumptions required for weak and strong convergence, see e.g. \cite{kloeden1992numerical}, do not hold.
\end{itemize}

Consequently, a considerable amount of research was devoted to the numerical
approximation of this equation, see
\cite{higham2005convergence,alfonsi2005discretization,kahl2006fast,berkaoui2007euler,lord-comparison,andersen-efficient},
to mention a few. However no strong convergence of order one results have been obtained so far up to best of our knowledge.
For a comparison of the different proposed schemes based on  simulation studies, see  
\cite{alfonsi2005discretization,lord-comparison}.

Our approach is based on a suitable transformation of the CIR process. Applying It\^{o}'s formula to $x(t)=\sqrt{y(t)}$ gives 
\begin{align} \label{eq:Volatility}
dx(t)  = \frac{1}{2}\k \left( \Big(\o - \frac{\s^{2}}{4\k} \Big) x(t)^{-1} - x(t)  \right)dt
            + \frac{1}{2}\s dw(t).
\end{align}
Zhu \cite{zhu2009modular} pointed out that
the drawback of the transformed equation is that the new mean level
$(\o - \frac{\s^{2}}{4\k}) x(t)^{-1}$ is stochastic and
that a naive Euler discretization cannot capture the erratic behavior
of $x(t)^{-1}$-term, although almost sure convergence of this method holds true, see \cite{jentzen2009pathwise}. The weakness of a naive Euler discretization is that its transition density is Gaussian and therefore its moments explode due to the $x(t)^{-1}$-term.

 On the other hand, Alfonsi showed  in \cite{alfonsi2005discretization} that the BEM applied to 
 \eqref{eq:orCIR}  preserves positivity of the solution and also monotonicity with respect to the initial
value. Moreover, his simulation studies indicated good convergence properties of this scheme.
In this paper we follow the recent result by Dereich et al. \cite{dereich2012euler},
where it was shown that the piecewise linear interpolation  of BEM applied to  \eqref{eq:orCIR},
strongly converges with a rate one half (up to a log-term).

\medskip

Given any step size $\D t>0$, the BEM scheme has the  form
\begin{align}\label{eq:CIRBEM}
X_{k+1}  & =  X_{k} + \frac{1}{2}\k \big( \o_{v} X_{k+1}^{-1} - X_{k+1}  \big )\D t
              +  \frac{1}{2}\s \D w_{k+1}, \qquad  k=0,1, \ldots, \\ X_0&=x(0) \nonumber
\end{align} with 
 $$\D w_{k+1} = w((k+1)\D t) - w(k\D t),  \qquad k=0,1,\ldots $$
and $$ \o_v =  \o - \frac{\s^{2}}{4\k}.$$
 We will establish a strong convergence of order one for the maximum 
$L^{p}$-distance in the discretization points between  \eqref{eq:orCIR}  and \eqref{eq:CIRBEM}, see Section 3.
For example for the $L^2$-distance we obtain
$$ \mathbb{E} \max_{k=0, \ldots, \lceil T/ \Delta t \rceil} |X_{k} - x(k \D t)|^{2} \leq C_2
\cdot \Delta t^2 \qquad \textrm{for}   \quad \frac{ 2\kappa \theta}{\sigma^2} > 3.$$ 
 As a consequence we also obtain  
the  same convergence order for the  approximation of the original CIR process
by $X_k^2$, $k=0,1,\ldots,\lceil T/\D t \rceil$.

\medskip
In this paper we will show that the above idea naturally extends to many types of SDEs 
with non-Lipschitz coefficients. The combination of the Lamperti transformation and the backward Euler scheme 
enables us to analyse the $L^p$-convergence rates for many scalar SDEs encountered in practice.  In particular, transforming BEM back
we obtain an order one scheme for the original SDE, which is close to a Milstein-type scheme,  see Section 4. 
Hence our approach turns out to be a new method for deriving numerical methods with
strong order one convergence for SDEs with non-Lipschitz coefficients.
Although strong convergence of  backward schemes for SDEs with non-Lipschitz coefficients has already been analysed in the literature and their convergence
for models as the Ait-Sahalia and the Heston $3/2$-volatility was obtained, see \cite{higham2012convergence,mao2012strong,szpruch-strongly},  schemes with strong convergence order one have not been established yet  in this setting. 

\medskip
Another motivation for our work are results by
Giles \cite{giles2006improved,giles2008multilevel}, who showed that for optimal MLMC simulations one should use discretization schemes
with strong convergence order one. Note that  strong convergence of the discretization scheme used for the       
MLMC simulations seems to be not only a sufficient but also a necessary condition \cite{hutzenthaler2011divergence}.

\medskip
The remainder of this paper is structured as follows. In the next section, we provide a general convergence result for
the BEM method applied to scalar SDEs with additive noise. Section 3 contains the results for our examples, i.e. the CIR, CEV,
Ait-Sahalia, 3/2-Heston volatility and Wright-Fischer SDEs. 
In Section 4, we provide the relation of BEM and a drift-implicit Milstein scheme and give an error analysis for the case of the CIR process. The last section
contains a short discussion.

\bigskip
\smallskip

\section{The BEM scheme for SDEs with additive noise}

\medskip

\subsection{Preliminaries}

Let $D=(l,r)$, where  $-\infty<l<r<\infty$, and let $a,b: D \rightarrow D$  be continuously differentiable functions. Moreover, let  $(\Omega, \mathcal{F},  (\mathcal{F}_t)_{t \geq 0}, \mathbb{P})$  be a filtered probability space and
$w(t)$, $t \geq 0$, a standard $(\mathcal{F}_t)_{t \geq 0}$-Brownian motion.
We begin with the SDE
\begin{align} \label{sde_target}
dy(t) = a(y(t))dt + b(y(t))dw(t), \quad  t \geq 0,\qquad  y(0) \in D,  
\end{align}
and assume that it has a unique strong solution with
$$ \mathbb{P}(y(t) \in D, \,\, t \geq 0 ) =1.$$
If $b(x)>0$ for all $x \in D$, then we can use the   Lamperti-type transformation 
\begin{equation*} 
F(x) = \lambda \int^{x} \frac{1}{b(y)}dy
\end{equation*}
for some $\lambda >0$. It\^{o}'s Lemma with $x(t)=F(y(t))$ gives the transformed SDE
\[
dx(t) = f(x(t)) dt +  \lambda dw(t),  \quad  t \geq 0,\qquad  x(0) \in F(D),
\]
with
$$ f(x)= \lambda \left( \frac{a(F^{-1}(x))}{b(F^{-1}(x))} -\frac{1}{2}b'(F^{-1}(x)) \right), \qquad x \in F(D), $$
where $F(D) =  (F(l),F(r))$.
Note that the classical Lamperti transformation corresponds to $\lambda=1$.
This transformation allows to  shift non-linearities from the diffusion coefficient
into the drift coefficient. Then, under appropriate assumptions on $f$ (respectively $a$ and $b$),
one can apply the backward Euler scheme
\begin{align} \label{BEM-lamp}
X_{k+1} = X_{k} + f(X_{k+1}) \D t +  \lambda \D w_{k+1}
\end{align}
and derive a strong convergence rate of order one for the maximum $L^p$-error in the discretization points, see Theorem \ref{th:mainTH} in the in Section 2.3.

\medskip
\subsection{Backward Euler-Maruyama scheme}

In this section we focus on the numerical approximation of
\begin{equation} \label{eq:SDE}
dx(t) = f(x(t))dt + \sigma \, dw(t), \quad t \geq 0, \qquad x(0)=x_0
\end{equation}
by the backward
Euler-Maruyama scheme
\begin{equation} \label{eq:BEM}
X_{k+1} = X_{k} + f(X_{k+1}) \D t +  \sigma \D w_{k+1},\quad  k=0,1, \ldots, \qquad X_0=x_0.
\end{equation}

\smallskip

We will work under the following assumption on the SDE itself:
\begin{assp} \label{as:0} 
Let $-\infty<\a<\be<\infty$ and assume that SDE \eqref{eq:SDE} has a unique strong solution which takes values in the set $(\alpha, \beta) \subset \mathbb{R}$, i.e.
\begin{align}\mathbb{P}(x(t) \in (\alpha, \beta), \,\, t \geq 0)=1.  \label{domain}\end{align}
\end{assp}

\smallskip

For the well-definedness of the drift-implicit Euler-Maruyama scheme we need the following assumption on the drift-coefficient:
\begin{assp} \label{as:A} The function $f: (\alpha, \beta) \rightarrow (\alpha, \beta)$ is continuous.
Moreover, there exists a constant $K \in \mathbb{R}$ such that
\begin{equation} \label{one_sided_lip}
 (x-y)(f(x)-f(y)) \le K\lev x-y \rev^2
\end{equation}
for all $x,y \in (\alpha, \beta)$.
\end{assp}

The Feller test, see e.g. Theorem V.5.29 in \cite{karatzas1991brownian}, gives that condition \eqref{domain} is equivalent to 
$$ \lim_{x \rightarrow \alpha_{+}} v(x) =  \infty, \qquad \qquad \lim_{x \rightarrow \beta_{-}} v(x) =   \infty $$
with 
$$ v(x)=  \frac{2}{\sigma^2} \int_{x_0}^x \int_{x_0}^{\zeta_2}  \frac{p'( \zeta_2)}{p'(\zeta_1)} d \zeta_1 d \zeta_2 $$
and
the  scale function 
$$p(x)= \int_{x_0}^x \exp \left( -2 \int_{x_0}^{\xi} \frac{f(u)}{\sigma^2} \, d u \right) \, d \xi, \qquad x \in (\alpha, \beta).$$
Note that $v$ can be rewritten as
$$ v(x)=  \frac{2}{\sigma^2} \int_{x_0}^x \int_{x_0}^{\zeta_2}  \exp  \left( -2 \int_{\zeta_1}^{\zeta_2} \frac{f(u)}{\sigma^2} \, d u \right)  d \zeta_1 d \zeta_2. $$
The condition on $v$ now implies  (recall that $\alpha < x_0 < \beta$)
\begin{align}
\label{exp_1}
 \lim_{x \rightarrow \alpha_{+}} \int^{x_0}_{x}  \exp  \left(2 \int_x^{\zeta_1} \frac{f(u)}{\sigma^2} \, d u \right)  d \zeta_1 =  \infty
\end{align}
and 
\begin{align} \label{exp_2}
 \lim_{x \rightarrow \beta_{-}} \int_{x_0}^{x}  \exp  \left( -2 \int_{\zeta_1}^{x} \frac{f(u)}{\sigma^2} \, d u \right)  d \zeta_1 =  \infty.
\end{align}
Now consider \eqref{exp_1} and assume that $ \limsup_{x \rightarrow \alpha_{+}} f(x) \neq  \infty$
However, if this would be true, the expression in \eqref{exp_1} would be finite due to the continuity of $f$.
Using a similar argument for \eqref{exp_2} we obtain
\begin{align} \limsup_{x \rightarrow \alpha_{+}} f(x) =  \infty, \qquad \qquad \liminf_{x \rightarrow \beta_{-}} f(x) =  - \infty.  \label{coerciv_help}
\end{align}

\smallskip

The drift-implicit Euler scheme is well defined if
\begin{equation*}
X_{k+1} - f(X_{k+1}) \D t = X_{k}  +  \sigma \D w_{k+1}
\end{equation*}
has a unique solution for $k=0,1, \ldots$. This is guaranteed by the following result:

\begin{lemma} \label{lem:Positivity}
 Let Assumption \ref{as:0} and \ref{as:A} hold and let $K\D t <1$. Moreover set
\begin{equation*} 
G(x)=x -  f(x)\D t, \quad x\in (\alpha, \beta).
\end{equation*}
Then for any $c \in \mathbb{R}$ there exists a unique $x \in (\alpha, \beta)$  such that
$G(x)=c$.
\end{lemma}
\begin{proof} The result follows if we can show that the
function $G$ is continuous, coercive and strictly monotone on $(\a,\be)$ (see
\cite{Zeidler1985}).  
However, due to Assumption  \ref{as:A} the function $G$ is continuous on $(\alpha, \beta)$.
Moreover, since
\begin{align*}
 (x-y)(G(x)-G(y)) &= (x-y)^2 - (x-y)(f(x)-f(y)) \D t  \\ & \ge (1 - K^+ \D t)(x-y)^{2} >0
\end{align*} by \eqref{one_sided_lip} (with  $K^+=\max \{0, K\}$)
the required strict monotonicity is obtained.
Finally, \eqref{coerciv_help} and the monotonicity imply that
$$\liminf_{x\rightarrow \alpha_{+} } G(x)=  \lim_{x\rightarrow \alpha_{+} } G(x)= -\infty$$
and $$\limsup_{x\rightarrow \beta_{-} } G(x)= \lim_{x\rightarrow \beta_{-} } G(x)=\infty,$$
so the function $G$ is coercive on $(\alpha, \beta)$.
\end{proof}
Note that for $K \leq 0$ there is no restriction on $\D t$. 

\smallskip
For completeness, we state here a well known discrete version of Gronwall's Lemma:

\begin{lemma} \label{gronwall_disc} Let $\D t >0$ and let $g_n, \lambda_{n} \in \mathbb{R}$, $n \in \mathbb{N}$, and $\eta \geq 0$ be given.
Moreover, assume that $1-\eta \D t >0$ and $1+\lambda_{n}>0$, $n\in\NN$. Then, if $a_n \in \mathbb{R}$, $n \in \mathbb{N}$, satisfies $a_0=0$ and
$$  a_{n+1} \leq a_n(1+ \lambda_{n}) + \eta a_{n+1} \D t + g_{n+1},  \qquad n=0,1, \ldots,$$ then this sequence also satisfies
$$  a_n \leq  \frac{1}{(1-\eta \D t)^n} \sum_{j=0}^{n-1} (1-\eta \D t)^j g_{j+1}\prod_{l=j+1}^{n-1}(1+\lambda_{l}), \qquad n=0,1, \ldots.  $$
\end{lemma}


\smallskip

Under the above assumptions we have the following moment bounds for the SDE and the BEM scheme:

\begin{lemma} \label{lem_sup_moments}
Let $T >0$ and let Assumption  \ref{as:0} and  \ref{as:A} hold. Then we have
\begin{align*} 
  \mathbb{E} \sup_{t \in [0,T]}
 |x(t)|^q  < \infty \end{align*}
for all $q \geq 1$.
If additionally $2K\D t < \eta$ for some $\eta<1$, then for all $q \geq 1$ there exist constants $C_{q}>0$, which are independent of $\D t$, such that 
\begin{align*}
\mathbb{E} \sup_{k=0, \ldots, \lceil T / \D t \rceil}
 |X_k|^q  \leq C_q.
\end{align*}

\end{lemma}
\begin{proof}
The first assertion can be shown by a straightforward modification of the proof  of Lemma 3.2 in \cite{higham2003strong}.

For the proof of the second assertion, we will denote constants which are independent of $\D t$ and whose particular value is not important by $c$ regardless of their
value.
Due to \eqref{coerciv_help} and the continuity of $f$ there exists an $x^{*} \in ( \alpha, \beta)$ with $f(x^*)=0$, so
we can rewrite  BEM
as  
\begin{equation*}
X_{k+1} - x^*  = X_k - x^* + (f(X_{k+1}) -f(x^*))\D t  +  \sigma \D w_{k+1}.
\end{equation*}
Multiplying with $X_{k+1}-x^*$ and using the one-sided Lipschitz condition on $f$ yield that
\begin{align*}
 (X_{k+1}-x^*)^2 &\leq \left((X_k -x^*) + \sigma \D w_{k+1} \right)(X_{k+1}-x^*) + K^+ \D t  (X_{k+1}-x^*)^2.
\end{align*}
where $K^+= \max \{0,K \}$.
Moreover, using $ab \leq \frac{1}{2}(a^2+b^2)$ and rearranging the terms gives
\begin{align*}
(X_{k+1}-x^*)^2 &\leq  2K^+ \D t (X_{k+1}-x^*)^2  +\left((X_k -x^*) + \sigma \D w_{k+1} \right)^2  
\end{align*}
  and hence, with $u_k=X_k-x^*$, we have
\begin{align*} 
u_{k+1}^2 &\leq u_{k}^2+ 2K^+ u_{k+1}^2\D t +   2 u_k \D w_{k+1} + \sigma^2 |\D w_{k+1}|^2.
\end{align*}
Note that $1-2K^+\D t \in (0,1]$
and
\begin{align} \label{bound_prop}  \sup_{ \D t \in (0, \eta /(2K^+))} \,  \sup_{k=1, \ldots, \lceil T / \D t \rceil} \, (1-2 K^+ \D t)^{-k}  < \infty.
\end{align}
So, the above discrete version of Gronwall's Lemma now yields
\begin{align} \label{key_rep}
u_{k}^2 \leq c & +  \frac{2}{(1-2 K^+ \D t)^{k}}\underbrace{\sum_{j=0}^{k-1} (1-2 K^+ \D t)^{j} u_j \D w_{j+1}}_{=M_k^{(1)}} \qquad  \\ & + 
\frac{1}{(1-2 K^+ \D t)^{k}}\underbrace{\sum_{j=0}^{k-1} (1-2 K^+ \D t)^{j} \sigma^2  \left( |\D w_{j+1}|^2 - \Delta \right)}_{=M_{k}^{(2)}}, \nonumber
\end{align} 
from which we obtain easily by induction that
 \begin{align*}    \sup_{k=0, \ldots, \lceil T / \Delta \rceil} \mathbb{E} u_{k}^2 < \infty.
  \end{align*}
Using this it can be easily checked that the processes $M_{k}^{(i)}$, $i=1,2$, are square-integrable martingales with respect to the filtration $\mathcal{F}_{k \D t}$, $k=0,1, \ldots$.
Hence Doob's inequality and straightforward calculations give for any $q \geq 1$ that
 \begin{align*}   \mathbb{E} \sup_{k=0, \ldots, \lceil T / \Delta \rceil} |M_{k}^{(1)}|^{2q}  \leq c \cdot 
 \sup_{k=0, \ldots, \lceil T / \Delta t \rceil} \mathbb{E} |u_{k}|^{2q}
  \end{align*}
and 
\begin{align*}   \mathbb{E} \sup_{k=0, \ldots, \lceil T / \Delta t \rceil} |M_{k}^{(2)}|^{2q} < \infty.
  \end{align*}
So using  \eqref{key_rep}  and \eqref{bound_prop} we obtain
\begin{align*}   \mathbb{E} \sup_{k=0, \ldots, \lceil T / \Delta t\rceil}  u_{k}^{4q}  \leq
c + c \cdot \sup_{k=0, \ldots, \lceil T / \Delta t \rceil}   \mathbb{E} u_{k}^{2q}
  \end{align*}
for $\D t < \eta /(2K^+)$
and the assertion  follows now by an induction argument in $q \in \mathbb{N}$.
\end{proof}

\medskip

\subsection{The Main Result}

Here we prove our general theorem on the  strong convergence of the numerical scheme \eqref{eq:BEM} 
to the solution of SDE \eqref{eq:SDE}.

\smallskip

\begin{assp} \label{as:MA}
Let $T >0$ and $p\ge2$. We assume that the drift coefficient  $f: (\alpha, \beta) \rightarrow (\alpha, \beta)$ of SDE \eqref{eq:SDE} 
is twice continuously differentiable and
satisfies
\[ 
\sup_{t \in [0,T]}
 \E  \lev  f'(x(t)) \rev^{p}  +
\sup_{t \in [0,T]} \mathbb{E} \left| (f'f)(x(t)) + \frac{\sigma^2}{2} f''(x(t)) \right|^{p}
< \8.  
\]
\end{assp}

\smallskip

\begin{theorem} \label{th:mainTH}
Let $T >0$, $p\ge 2 $, $\eta \in (0,1)$ and  Assumptions \ref{as:0}, \ref{as:A} and \ref{as:MA} hold. 
Then, for $2K\D t < \eta$, there exists a constant $C_{p}>0$ (independent of $\D t$) such that 
\begin{equation}  \label{eq:Tconv}
\E\left[  \sup_{k=0, \ldots ,\lceil T / \Delta t \rceil}  \lev x(k \D t) -  X_{k}\rev^{p} \right] \le C_{p}  \cdot \D t^{p} .
\end{equation}
\end{theorem}

\begin{proof}

Recall that we will denote constants which are independent of $\D t$ and whose particular value is not important by $c$ regardless of their
value.
By applying It\^{o}'s formula  on $f(x(t))$ we have
\begin{align} \label{eq:Taylor}
 x((k+1)\D t) = x(k\D t)  & + \int_{k\D t}^{(k+1)\D t} f(x((k+1) \D t)) dt \\ & \nonumber + \s \int_{k\D t}^{(k+1)\D t} dw(t) + R_{k+1}
\end{align} 
where 
\begin{align} \label{remainder}
 R_{k+1} = & - \int_{k\D t}^{(k+1)\D t} 
\int_{t}^{(k+1)\D t}  \left( (f'f)(x(s))  +  \frac{\s^{2}}{2}f''(x(s)) \right) ds dt
\\ &  \qquad 
 - \s \int_{k\D t}^{(k+1)\D t}\int_{t}^{(k+1)\D t} f'(x(s))  dw(s)  dt. \nonumber
\end{align}
We can decompose
$R_{s}$  as
$R_{s} = R_{s}^{(1)} + R_{s}^{(2)}$ with  
\begin{align*} \
 R_{s}^{(1)} &=  - \int_{(s-1) \D t}^{ s \D t} 
\int_{t}^{s \D t}  \left( (f'f)(x(u))  +  \frac{\s^{2}}{2}f''(x(u)) \right) du dt \\
 R_{s}^{(2)} &=- \s \int_{(s-1) \D t }^{ s \D t }\int_{t}^{s \D t} f'(x(u))  dw(u)  dt. 
\end{align*}

Using equations \eqref{eq:Taylor} and \eqref{eq:BEM} we   
have
\begin{align*}
 x((k+1)\D t) - X_{k+1}
 = x(k\D t) -  X_{k}  +  [ f( x((k+1)\D t)) - f (X_{k+1}) ] \D t 
             + R_{k+1}
\end{align*}
and thus
\begin{align*}
\bigl( x((k+1)\D t) - X_{k+1} - [ f( x((k+1) \D t )) - f (X_{k+1}) ] \D t \bigr)^{2}
 = \bigl( x(k \D t ) -  X_{k} + R_{k+1} \bigr)^{2}. 
\end{align*}
We arrive at
\begin{align*}
 & \lev x((k+1)\D t)  - X_{k+1} \rev^{2}  -  \lev x(k\D t) -  X_{k} \rev^{2} \\ 
 &  \quad =   2 ( x((k+1)\D t) - X_{k+1} )  [ f( x((k+1)\D t)) - f (X_{k+1}) ] \D t \\
& \qquad -  [ f( x((k+1)\D t)) - f (X_{k+1}) ]^2  \D t^{2} 
   + 2 ( x(k\D t) -  X_{k} ) R_{k+1}   + R_{k+1}^{2}. 
\end{align*}
Using the one-sided Lipschitz assumption on $f$ we obtain 
\begin{align*}
(1-2K^+ \D t) \lev x((k+1)\D t) - X_{k+1} \rev^{2} 
 \le \lev x(k\D t) -  X_{k} \rev^{2} + 2 ( x(k\D t) -  X_{k} ) R_{k+1}   + R_{k+1}^{2}.
\end{align*}
Let us define $e_{k} =  x(k\D t) - X_{k} $ and $\gamma_{ \D t}= (1-2K^+ \D t)$.
Note that $\gamma_{ \D t} \in (0,1]$
and
\begin{align} \label{help_ga}  \sup_{ \D t \in (0, \eta/(2K^+))} \,  \sup_{k=1, \ldots, \lceil T / \D t \rceil} \, \gamma_{ \D t}^{-k}  < \infty.
\end{align}
Now, Lemma \ref{gronwall_disc} yields
\begin{align} \label{key_rep_2}
e_{k}^2 \leq 2 \sum_{s=0}^{k-1} \gamma_{ \D t}^{s-k} e_s R_{s+1} + \sum_{s=0}^{k-1} \gamma_{ \D t}^{s-k}  R_{s+1}^2.
\end{align}

Since $\mathbb{E} \big[ R_{s+1}^{(2)}  \big | \mathcal{F}_{s \D t} \big] =0$, we have that
$$ \sum_{s=0}^{k-1}\g^{s}_{\D t} e_{s} R^{(2)}_{s+1}, \qquad k=0, \ldots, \lceil T / \Delta t \rceil,$$
is a martingale  and the Burkholder-Davis-Gundy inequality implies that
$$ \mathbb{E}  \left[ \sup_{k=0, \ldots , \ell }  \left|\sum_{s=0}^{k-1}\g^{s}_{\Delta t} e_{s} R^{(2)}_{s+1}  \right|^{q} \right]\leq c \,  
\E \left( \sum_{s=0}^{\ell -1 }  \g^{2s}_{\Delta t} \left| e_{s} \right|^2  \left| R^{(2)}_{s+1}  \right|^2 \right)^{q/2}
$$ for any $q \geq 1$ and $\ell=0, \ldots, \lceil T / \Delta t \rceil$.
So using the boundedness of $\gamma_{\D t}$ and \eqref{help_ga}
we have
$$ \mathbb{E}  \left[ \sup_{k=0, \ldots , \ell}  \left|\sum_{s=0}^{k-1}\g^{s-k}_{\Delta t} e_{s} R^{(2)}_{s+1}  \right|^{q} \right]\leq c \,  
\E \left( \sum_{s=0}^{\ell-1}  \left| e_{s} \right|^2  \left| R^{(2)}_{s+1}  \right|^2 \right)^{q/2}
$$ for any $q \geq 1$.
Using this and Jensen's inequality in \eqref{key_rep_2} we now arrive at 
\begin{align} \label{prae_gronwall}
 \E \left[ \sup_{k=0, \ldots ,\ell}| e_{k}|^{2q} \right] 
\le & c\, \lceil T / \Delta t \rceil^{q/2-1}  \sum_{s=0}^{ \ell -1} \E  | e_{s}|^q  \left| R^{(2)}_{s+1}  \right|^q \\
& + c\,\lceil T / \Delta t \rceil^{q-1} \sum_{s=0}^{\ell-1} \E |e_{s} |^{q} \left|R^{(1)}_{s+1} \right|^{q} \nonumber\\
& + c \,\lceil T / \Delta t \rceil^{q-1}   \sum_{s=0}^{\ell-1} \E |R_{s+1}|^{2q} . \nonumber
\end{align}
Now, Assumption \ref{as:MA},  Jensen's inequality and the Burkholder-Davis-Gundy inequality give that
\begin{align} \label{loc_est_3}
  \E \left[ \left|R_{k+1}^{(2)}\right|^m \Big{|} \, \F_{ k \Delta t} \right]   \leq  c \left(1 +|f'(x_{k \Delta})|^m \right) \D t^{3m/2}
\end{align}
and
\begin{align}
 \E \left[ \left|R_{k+1}^{(1)}\right|^m \Big{|} \, \F_{ k \Delta t} \right]    & \leq  \label{loc_est_4}
c \, \left( 1+ \left|(f'f)(x(k \Delta)) + \frac{\s^{2}}{2}f''(x(k \Delta))\right|^{m} \right) \D t^{2m} 
\end{align}
for all $m \leq p$. Thus, the Cauchy-Schwarz inequality and Assumption \ref{as:MA} yield that
$$ \E  | e_{s}|^q  \left| R^{(2)}_{s+1}  \right|^q = \E \left[ |e_{s}|^q \, \mathbb{E} \left[  \left| R^{(2)}_{s+1}  \right|^q \Big{|} \, \mathcal{F}_{ s \D  t} \right] \right]
\leq c \left(1+ \mathbb{E} | e_{s}|^{2q} \right)^{1/2}\D t^{3q/2}.
 $$
Hence Young's inequality implies
$$ 
c\, \lceil T / \Delta t \rceil^{q/2-1}  \sum_{s=0}^{ \ell -1} \E  | e_{s} |^q  \left| R^{(2)}_{s+1}  \right|^q 
\leq c  \sum_{s=0}^{ \ell -1} \E  \left| e_{s} \right|^{2q} \D t +  c \D t^{2q}. 
$$
Similar we also obtain
$$ c\,\lceil T / \Delta t \rceil^{q-1} \sum_{s=0}^{\ell-1} \E |e_{s} |^{q} \left|R^{(1)}_{s+1} \right|^{q}
\leq c  \sum_{s=0}^{ \ell -1} \E  \left| e_{s} \right|^{2q} \D t +  c \D t^{2q}.$$
Since finally
$$c \,\lceil T / \Delta t \rceil^{q-1}   \sum_{s=0}^{\ell-1} \E |R_{s+1}|^{2q}  \leq  c  \D t^{2q},$$
by inserting these three estimates in \eqref{prae_gronwall} we end up with
\begin{align*}
 \E \left[ \sup_{k=0, \ldots ,\ell}| e_{k}|^{2q} \right] 
\le c  \sum_{s=0}^{ \ell -1} \E  \left| e_{s} \right|^{2q} \D t  + c  \D t^{2q}
\end{align*}
and Gronwall's Lemma completes the proof.

\end{proof}

\smallskip

The above result and Lemma \ref{lem_sup_moments}
now give convergence (without a rate) in all $L^q$-norms:

\smallskip

\begin{coro} Under the assumptions of Theorem \ref{th:mainTH} we have
\begin{equation}  
\lim_{\D t \rightarrow 0} \, \E\left[  \sup_{k=0, \ldots ,\lceil T / \Delta t \rceil}  \lev x(k \D t) -  X_{k}\rev^{q} \right] =0
\end{equation}
 for all $q \geq 1$.
\end{coro}

\medskip
\begin{rmk}\label{comp_alf_1}
 In \cite{alfonsi2012}, independently of the research in this paper, a similar result to Theorem \ref{th:mainTH} is established for the case 
$D=(\alpha, \infty)$ and drift functions $f:D \rightarrow \mathbb{R}$, which are twice continuously
differentiable and satisfy a monotone condition (which is equivalent to our one-sided Lipschitz condition).
Using a continuous extension of BEM Alfonsi obtains the error bound \eqref{eq:Tconv}
under the assumption
\[ 
\E \left( \int_{0}^T
  \lev  f'(x(t)) \rev^{2}\, dt \right)^{p/2}  +  \mathbb{E} \left( \int_0^T \left| (f'f)(x(t)) + \frac{\sigma^2}{2} f''(x(t)) \right|  dt \right)^{p}
< \8
\]
for $p \geq 1$. This result is then applied to the CIR and CEV process, i.e.
 Propositions \ref{subsec_cir} and \ref{subsec_cev} are obtained.

Note that due to our  bound on the inverse moments on the LBE, see the subsection
below, we are also able to  cover SDEs like the the Heston 3/2-volatility and the Ait-Sahlia model.
Moreover, since we work under the assumption $D=(\alpha, \beta)$ we can also treat  the Wright-Fisher SDE and similar equations.
\end{rmk}

\medskip

\subsection{Boundedness of inverse moments of BEM}

If the drift coefficient has an even more specific structure, see the assumption below, we can also control the inverse moments of BEM. For the Heston-3/2 volatility and also the Ait-Sahalia model this will be helpful later on.

\smallskip

\begin{assp} \label{as:INV} Let $\alpha \geq  0$ and assume that the drift coefficient $f:(\alpha, \beta) \rightarrow (\alpha, \beta)$ has the structure
$$ f(x)= \frac{c_1}{x^{m_1}} + h(x) , \qquad x \in  (\alpha, \beta) $$
where 
\[
 \lev h(x) \rev \le c_2 \cdot (1 +\lev x\rev^{m_2} )\qquad x \in  (\alpha, \beta)
\]
for some $c_1,c_2>0$ and $m_1, m_2 >0$.
\end{assp}

\smallskip

Under the above assumption BEM can be written as 
$$ X_{k+1} = X_{k} + (  c_1 X_{k+1}^{-m_1} + h(X_{k+1}) )\D t + \s \D w_{k+1}  $$
and we have
\begin{align}  \label{inv_trick_1}
 \frac{1}{X_{k+1}^{m_1}} &= \frac{1}{c_1 \D t} \left(  X_{k+1} -X_k - h(X_{k+1}) \D t - \s \D w_{k+1} \right).   
\end{align}
Proceeding as in the proof of Theorem \ref{th:mainTH} we also have
\begin{align} \label{inv_trick_2}
\frac{1}{x(t_{k+1})^{m_1}} &= \frac{1}{c_1 \D t} \left(  x(t_{k+1}) -x(t_k) - h(x(t_{k+1})) \D t - \s \D w_{k+1} -R_{k+1}\right),   
\end{align}
with $R_{k+1}$ given by \eqref{remainder}. This can be used to derive the following result:

\smallskip

\begin{lemma} \label{bound_inv}
 Let $T >0$ and $p \geq 2$. Moreover, let the assumptions of Theorem \ref{th:mainTH} and  let also Assumption \ref{as:INV} hold.
Then there exists constants $C_p^{(1)},C_p^{(2)} >0$ such that we have
\begin{align*}
 \sup_{k=0, \ldots, \lceil T / \D t \rceil } \mathbb{E}  
| X_{k}|^{-m_1p} & < C_p^{(1)} \left( 1 +  \sup_{t \in [0,T] }   \mathbb{E}
| x(t)|^{-m_1 p} \right).
\end{align*}
and
\begin{align*}
\mathbb{E} \sup_{k=0, \ldots, \lceil T / \D t \rceil } 
| X_{k}|^{-m_1p} & < C_p^{(2)} \left( 1 +  \mathbb{E} \sup_{t \in [0,T] } 
| x(t)|^{-m_1 p} \right).
\end{align*}
\end{lemma}

\smallskip

\begin{proof} We only prove the second assertion, the proof of the first assertion is similar.
Using \eqref{inv_trick_1}, \eqref{inv_trick_2}, \eqref{eq:Tconv}, \eqref{loc_est_3} and \eqref{loc_est_4}  we have
\begin{align*}
\mathbb{E} \sup_{k=0, \ldots, \lceil T / \D t \rceil } 
|x(t_k)^{-m_1}- X_{k}^{-m_1}|^{p} & \leq  c \left(1  +  \mathbb{E} \sup_{k=0, \ldots, \lceil T / \D t \rceil } |h(x(t_k)) - h(X_k) |^p \right).
\end{align*}
Since
$$ |h(x)-h(y)| \leq |h(x)| + |h(y)| \leq  c_2 (1+|x|^{m_2} +|y|^{m_2}), $$ Lemma \ref{lem_sup_moments} 
and the triangle inequality    finish the proof.
\end{proof}

\smallskip

In the next Lemma we establish a general a-priori estimate for uniform inverse moments of SDE \eqref{eq:SDE}
with the drift structure imposed by Assumption \ref{as:INV}. 
\begin{lemma} \label{lem:INV_SDE}
 Let $p \ge 2$. Let the assumptions of Lemma \ref{lem_sup_moments} hold and in addition let the drift of SDE \eqref{eq:SDE} satisfy Assumption \ref{as:INV}. Then there exists a 
constant $C_{p}>0$ such that we have
$$
\mathbb{E} \sup_{t \in [0,T] } 
| x(t)|^{- p} \le C_{p}\left( 1 +   \sup_{t \in [0,T] } \mathbb{E} | x(t)|^{- (p+2)}  \right). 
$$ 
\end{lemma}
\begin{proof}
Let ${\{\a_{n}\}}_{n\in \NN }$ and  ${\{\be_{n}\}}_{n\in \NN }$ be such that $\a_{n} \searrow\a $ and $\be_{n}\nearrow \be $, when $n \rightarrow \8$. Let us choose $n_{0}$ such that $\a_{n_{0}} \le x(0) \le \be_{n_{0}}$. Then for $n\ge n_{0}$, we define the stopping time $\t_{n} = \{ t>0: x(t)\notin (\a_{n},\be_{n})\}$.  By It\^{o}'s lemma we have
\begin{equation*}
 \begin{split}
  |x(t\wedge \t_{n})|^{-p}   = |x(0)|^{-p} & + \int_{0}^{t} \left( g(x(s)) -p | x(s)|^{-(p+1)}h(x(s))  \right)\1_{ [0,\t_{n}]}(s) ds \\
& -p \s \int_{0}^{t} | x(s)|^{-(p+1)}\1_{ [0,\t_{n}]}(s) dw(s),
 \end{split}
\end{equation*}
where $g(x) = - c_{1} p | x|^{-(p+1+m_{1})} + \fracs{\s^{2}p(p+1)}{2} 
|x|^{-(p+2)}$.  Observe that for $m_{1}>1$, $\lim_{x\searrow 0}g(x) = -\8$ and $\lim_{x\rightarrow \8}g(x) = 0$,
hence in that case there exists a $c>0$ such that $$\sup_{x>0}g(x)\le c.$$ If $m_{1}\in(0,1]$ then there exists a  
$c>0$ such that
$$
| g(x) | \le c(1 + |x|^{-(p+2)}  ).
$$

By Assumption \ref{as:INV} and Burkholder-Davis-Gundy's inequality we have
\begin{align*}
  \E \left[ \sup_{t\in[0,T]}|x(t\wedge \t_{n})|^{-p} \right]  & \le |x(0)|^{-p} 
+ \int_{0}^{T} c \left( 1 +  \E |x(s)|^{-(p+2)} + \E | x(s)|^{-(p+1)+m_{2}}  \right)ds \\
& \qquad +c \,\E \left( \int_{0}^{T} | x(s)|^{-2(p+1)}\1_{[0,\t_{n}]}(s) ds \right)^{1/2}\\ 
& \le |x(0)|^{-p} 
+ \int_{0}^{T} c \left( 1 +  \E |x(s)|^{-(p+2)}   \right)ds \\
& \qquad +c \,\E \left( \sup_{t\in[0,T]} |x(t\wedge \t_{n}) |^{-p} \int_{0}^{T} | x(s)|^{-(p+2)} ds \right)^{1/2}.
\end{align*}
Applying Young's inequality to the last summand of the above inequality now yields
$$ 
\E \left[ \sup_{t\in[0,T]}|x(t\wedge \t_{n})|^{-p} \right]   \le c\left(1 + \sup_{t \in [0,T] } \mathbb{E} | x(t)|^{- (p+2)}   \right).
$$  
and Fatou's Lemma completes the proof.

\end{proof}

\bigskip
\smallskip

\section{Examples}
 
In this section, we will apply our main result to several examples. To simplify the presentation, we
will  denote the numerical method $Y_k=F^{-1}(X_k)$, $k=0,1, \ldots$,
where  $X_k$, $k=0,1, \ldots$ is given by \eqref{BEM-lamp},
as
{\it Lamperti-backward Euler} (LBE) approximation of SDE \eqref{sde_target}.
Moreover, we will say that this method is {\it $p$-strongly convergent with order one}, if 
\begin{equation*} 
\E  \sup_{k=0, \ldots ,\lceil T / \Delta t \rceil}  \lev y(t_{k}) -Y_k  \rev^{p} \le C_p \cdot \D t^{p} .
\end{equation*} 
Finally, constants whose particular value is not important will be again denoted by $c$.

\medskip

\subsection{CIR process}\label{sub_cir}
Recall that the Cox-Ingersoll-Ross process is given by the SDE
\begin{equation} \label{cir_in_sub_cir}
dy(t) = \k(\o-y(t))dt + \s\sqrt{y(t)}d w(t), \quad t \geq 0, \qquad y(0)>0.
\end{equation}
If $2\k \o \ge\s^2$, then we have $D=(0, \infty)$ and Assumption \eqref{as:0} holds for
$(\a,\be)=(0,\8)$. Moreover, recall that
the transformed SDE using $F(y)=\sqrt{y}$ reads as
\begin{align} \label{eq:Volatility}
dx(t)  = f(x(t))dt + \frac{1}{2}\s dw(t), \quad t \geq 0, \qquad x(0)=\sqrt{y(0)}
\end{align}
with
$$f(x)=\frac{1}{2}\k \left( \Big(\o - \frac{\s^{2}}{4\k} \Big) x^{-1} - x  \right), \qquad x>0$$
and the BEM scheme is given by
\begin{align} \label{eq:CIRBEM_2}
X_{k+1}  =  X_k + f(X_{k+1})\Delta t +\frac{1}{2}\s \Delta w_{k+1}, \qquad k=0,1, \ldots
\end{align}
with $X_0=x(0)$.
Straightforward calculations give 
\begin{align*}
 (x-y)(f(x)-f(y)) \leq   - \frac{1}{2}\k(x - y)^{2}, \qquad x,y >0,
\end{align*}
so  Assumption \ref{as:A} holds with $K=-\kappa/2$. 
Observe also that
\[
f'(x) = -\frac{1}{2}\k  (\o_{v} x^{-2} + 1 ),
\]
and
\[
  (f'f)(x) +  \frac{\s^{2}}{2}f''(x)  =  - \frac{\k^2}{4} ( \o_{v}^2 x^{-3} - x )
                                  + \frac{1}{2}\k  \o_{v}x^{-3}\s^2 .
\]
So, for Assumption \ref{as:MA} to hold we need
\[
 \sup_{0\le t \le T}\E[x(t)^{-3p}] =  \sup_{0\le t \le T}\E[y(t)^{-\frac{3}{2}p}].
\]
Since
\begin{align} \label{cir_trans_inv}
 \sup_{0\le t \le T}\E [y(t)^{q}] < \8 \quad \text{for} \quad q>-\frac{2k\o}{\s^{2}},
\end{align}
see e.g. \cite{dereich2012euler}, Assumption \ref{as:MA} and  as a consequence Theorem \ref{th:mainTH} hold 
if $p<\frac{4}{3}\frac{k\o}{\s^{2}}$.
 
\medskip

In order to  approximate the original CIR process observe that 
\[
 ( x(t_{k})^{2} -  X_{t_{k}}^{2} ) = (x(t_{k}) +  X_{t_{k}})(x(t_{k}) -  X_{t_{k}}).
\]
Let $\e>0$ such that $p(1+\e)<\frac{4}{3}\frac{k\o}{\s^{2}}$. Then H\"{o}lder's inequality  gives
\begin{equation*}
 \begin{split}
  & \E  \left[  \sup_{0\le k \le \lceil T / \Delta t \rceil }\lev x(t_{k})^{2}  -  X_{t_{k}}^{2} \rev^{p}\right] \\
& \le  \left(\E \left[ \sup_{0\le k \le \lceil T / \Delta t \rceil }\lev x(t_{k}) +  X_{t_{k}} \rev^{p\frac{1+\e}{\e}} \right]\right)^{p\frac{\e}{1+\e}} 
    \left(\E \left[ \sup_{0\le k \le \lceil T / \Delta t \rceil }\lev x(t_{k}) -  X_{t_{k}} \rev^{p(1+\e)} \right]\right)^{\frac{1}{1+\e}}.
 \end{split}
\end{equation*}
Using Lemma \ref{lem_sup_moments} we obtain:

\smallskip
\begin{prop} \label{subsec_cir} Let $T>0$ and $2 \leq p<\frac{4}{3}\frac{\k\o}{\s^{2}}$. Then, 
the LBE approximation of the CIR process is $p$-strongly convergent with order one.
\end{prop}

\smallskip

\subsection{Numerical Experiment}
Note that the unique solution to  \eqref{eq:CIRBEM_2} is given by
\[
 X_{k+1} =  \frac{1}{2+ \k\D t } \left( X_{k} + \frac{1}{2}\s \D w_{k+1}
    + \sqrt{       \left(X_{k} + \frac{1}{2}\s \D w_{k+1} \right)^2
     +   \k  \o_{v} \D t     } \right)
\]
with $ \o_v =  \o - \frac{\s^{2}}{4\k}.$
Hence implicit schemes not necessarily increase the computational complexity with comparison to classical explicit
procedures. 
\newline
In our numerical experiment, we focus on
the $L^2$-error at the endpoint $T=1$, so we let
\begin{equation*}
e_{\D t}=\E \lev x(T)-X_{T}\rev^2.
\end{equation*}
For our numerical experiment we set $\o=0.125$, $\k=2$, and $\s=0.5$. This gives $\fracs{2\k\o}{\s^2}=2$ and corresponds to the critical parameters for which Dereich et al. \cite{dereich2012euler} established strong convergence of order one half for linearly interpolated BEM \eqref{eq:CIRBEM_2} with respect to the uniform $L^2$-error criteria. Although, theoretical results obtained in this paper impose slightly more restrictive conditions for parameters than those in Dereich et al., performed numerical experiment suggests that for practical simulations condition $1< p <\fracs{2\k\o}{\s^2}$ suffices for uniform-$L^p$  convergence with order one. 
Although an explicit solution to
\eqref{eq:Volatility} is unknown, Theorem \ref{th:mainTH} guarantees
that BEM  strongly converges to the true solution. Therefore, it
is reasonable to take BEM with a very small time step, we choose $\D t = 2^{-15}$,
as a reference solution. We then compare it to BEM evaluated with
$( 2^{4}\D t, 2^{5}\D t, 2^{6}\D t, 2^{7}\D t)$
in order to estimate the rate of the $L^2$-convergence, where we estimate $e_{\D t}$ by a Monte-Carlo procedure, i.e.
$$ e_{\D t} \approx \frac{1}{10^3} \sum_{i=1}^{10^3}  \lev x^{(i)}(T)-X^{(i)}_{T}\rev^2. $$
Here $x^{(i)}(T),X^{(i)}_{T}$ are iid copies of $x(T),X_{T}$.
We plot $e_{\D t}$ against $\D t$ on a log-log scale, i.e.
if we assume
that a power law relation $e_{\D t} = C \D t^{q}$ holds for some
constant $C$ and $q$, then we have $\log e_{\D t} = \log C + q \log \D t$.
For our simulation, 
a least squares fit for $\log C$ and $q$ yields the value $1.9332$ for $q$
with a least square residual of $0.016$.  Hence, our results are consistent
with strong order of convergence equal to one.

\begin{figure}[htb]
\begin{center}
 \includegraphics[scale=0.5]{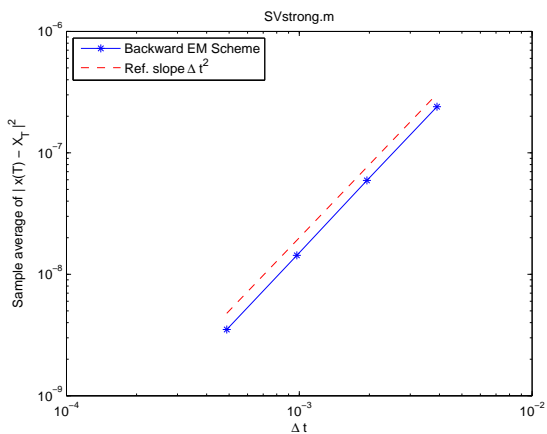}
\end{center}
        \caption[bif]{Strong error plot for backward Euler scheme applied to CIR process. 
        }
        \label{fig:strong}
\end{figure}

\medskip

\subsection{Heston 3/2-volatility}

In \cite{heston1997simple} the inverse of a CIR process is used as a  stochastic volatility process, which gives the
so-called Heston-$3/2$-volatility 
\begin{align}
dy(t) &=  c_1 y(t)( c_2 -  y(t)) \, dt +  c_3 y(t)^{3/2}  \, dw(t), & \quad t \geq 0, \qquad y(0)>  0
\label{3/2_cir}
\end{align} where $c_1,c_2,c_3 >0$.
Using $F(y)=y^{-1/2}$  leads to
\begin{align}
 \label{3/2_cir_trans}
dx(t)= \left( \left( \frac{c_1}{2} + \frac{3c_3^2}{8}\right) x(t)^{-1}  - \frac{c_1c_2}{2} x(t)\right) \,dt - \frac{c_3}{2} \, d w(t),
\end{align}
which coincides with 
\eqref{eq:Volatility} if we use a reflected Brownian motion, i.e. $-w$, which is still a Brownian motion,
and
$$ \sigma=c_3, \qquad \theta= \frac{1}{c_2} + \frac{c_3^2}{c_1c_2}, \qquad \kappa= c_1c_2. $$
Hence we have the relation
$$ \frac{ \kappa \theta}{\sigma^2} = 1+\frac{c_1}{c_3^2}, $$
so Theorem \ref{th:mainTH} holds here for $p<\fracs{4}{3}(1+\frac{c_1}{c_3^2})$.
Note that  the Heston-$3/2$-volatility is one of the SDEs which does not have finite moments of any order. 
As the inverse of the CIR process it has finite $q$-moments up to order $q<2 + \frac{2c_1}{c_3^2}$.

Now, for  transforming back we have to control the inverse moments of the BEM scheme for CIR. Here Lemma \ref{bound_inv} and \ref{lem:INV_SDE} give 
\begin{align*}
\mathbb{E} \sup_{k=0, \ldots, \lceil T / \D t \rceil } 
| X_{k}|^{-p} &  \le C_{p} \left( 1 +   \sup_{t \in [0,T] } \mathbb{E} | x(t)|^{- (p+2)}  \right)
\end{align*}
for $p<\fracs{4}{3}(1+\frac{c_1}{c_3^2})$.
From the analysis of the CIR process we have that 
$$
\sup_{t \in [0,T]} \mathbb{E} | x(t)|^{-(p+2)} < \8 \quad \text{for} \quad p< 4 \left( \frac{1}{2}+ \frac{c_1}{c_3^2} \right).
$$
To establish  the convergence result  for the LBM for the Heston-3/2 volatility note that
$$ \left| \frac{1}{X_k^2} - \frac{1}{x(t_k)^2} \right| = \frac{X_k + x(t_k)}{X_k^2 x(t_k)^2} |X_k - x(t_k)| \leq  
c( |X_k|^{-3} +  |x(t_k)|^{-3}) |X_k - x(t_k)|.
$$
Finally using H\"older's inequality 
gives
 \begin{align*}
& \E \left[\sup_{k=0, \ldots, \lceil T / \D t \rceil } \left| \frac{1}{X_k^2} - \frac{1}{x(t_k)^2} \right|^p \right] 
\\ & \qquad \qquad \leq  
c \left(  \left( \E  \sup_{k=0, \ldots, \lceil T / \D t \rceil }|X_k|^{-4p} \right)^{3/4} +  \left( \E \sup_{k=0, \ldots, \lceil T / \D t \rceil } |x(t_k)|^{-4} \right)^{3/4} \right)\\ 
\\ & \qquad \qquad \qquad  \times \left(\E \sup_{k=0, \ldots, \lceil T / \D t \rceil } |X_k - x(t_k)|^{4p} \right)^{1/4}.
\end{align*}
Hence we obtain:
\medskip
\begin{prop}
If  $ 1 \leq p < \fracs{1}{3} +\fracs{1}{3}\frac{c_1}{c_3^2}$, then the LBE approximation of the Heston-3/2 process is $p$-strongly convergent
with order one.
\end{prop}

\smallskip

\subsection{CEV  process}

Another popular SDE in finance is the mean reverting constant elasticity of variance process (\cite{cox1975notes}) given by
\begin{equation} \label{eq:SABR}
dy(t) = \k(\o-y(t))dt + \s y(t)^{\a}d w(t) \\
\end{equation}
where $0.5 < \a<1$, $\k,\o,\s >0$. 
By the Feller test  we have $D=(0, \infty)$ .
Applying It\^{o}'s formula to the function $F(y(t))=y(t)^{1-\a}$
we obtain
that Assumptions \ref{as:0} holds with $(\a,\be)=(0,\8)$ and
\begin{equation} \label{eq:THeston}
d x(t)  =   f(x(t))dt
              + (1-\a) \s dw(t) \\
\end{equation}
with
$$f(x)= (1-\a) \left(  \k \o x^{-\frac{\a}{1-\a}} - \k  x
               - \frac{\a \s^{2}}{2} x^{-1} \right), \qquad x >0.$$
Again we need to check the remaining assumptions.
Since $\a>0.5$, and consequently $\frac{1}{1-\a}>2$, we have for
\[
  f'(x) = -\a   \k \o x^{-\frac{1}{1-\a}} - (1-\a) \k  + (1-\a) \frac{\a \s^{2}}{2} x^{-2}, \qquad x>0
\]
that
$$ \lim_{x \rightarrow 0} f'(x)= - \infty, \qquad \lim_{x \rightarrow \infty } f'(x) = -(1- \a) \k$$
and hence there exists a $c>0$ such that
$$ \sup_{x >0} f'(x) \leq c.$$
Now the mean value theorem implies
\[
 (x-y) (f(x) - f(y)) \le c \lev x- y\rev^{2}, \qquad x,y>0,
\] i.e. the drift coefficient is one-sided Lipschitz.
We have moreover that
\[
 f''(x) =  \frac{\a}{1-\a}  \k\o x^{-\frac{2-\a}{1-\a}}  - (1-\a) \a \s^{2} x^{-3} .
\]
However, from \cite{berkaoui2007euler} it is known that
\begin{equation}  \label{moments_cev}
 \sup_{0\le t \leq T} \E\left\vert y(t)\right\vert
 ^p< \8 \qquad \textrm{for any} \quad p \in \mathbb{R},\,\, T>0,
\end{equation} and 
therefore also Assumption \ref{as:MA} holds and Theorem \ref{th:mainTH} can be applied for any $p \geq 1$.
For the back transformation note that the mean value theorem yields
\[
 | (x^{\frac{1}{1-\a}}-y^{\frac{1}{1-\a}}) | \le \frac{1}{1-\a}( x^{\frac{\a}{1-\a}} + y^{\frac{\a}{1-\a}}) | x-y |, 
\]
Using Lemma \ref{lem_sup_moments}  we have:

\smallskip

\begin{prop}\label{subsec_cev}
Let $p \geq 1$. The LBE approximation of the CEV process is $p$-strongly convergent
with order one.
\end{prop}

\smallskip

\subsection{Wright-Fisher Diffusion}
The Wright-Fisher SDE that originated from mathematical biology,\cite{either1986markov} and recently is also gaining popularity in 
mathematical finance \cite{howison2010risk,larsen2007diffusion} reads as
\begin{align*}
dy(t)=(a-b y(t))dt + \gamma \sqrt{|y(t)(1-y(t))|}\, dw(t), \qquad t \geq 0, \qquad y_0 \in (0,1)
\end{align*}
with $a,b,\gamma>0$. If
\begin{align} \label{wf_1} \frac{2a}{\gamma^2} \geq 1 \qquad \textrm{and} \qquad \frac{2(b- a)}{\gamma^2} \geq 1, \end{align}
then this SDE has a unique strong solution with
$$ \mathbb{P}(y(t) \in (0,1), \,\, t \geq 0 ) =1,$$
see \cite{karlin1981}. 
Using 
$$ {F}(y)= 2 \arcsin(\sqrt{y}), \qquad y \in (0,1)$$
we obtain 
$$  dx(t) = f(x(t))dt + \gamma\, dw(t), \qquad t \geq 0, \qquad x_0 = 2 \arcsin(\sqrt{y_0})$$
with
$$ f(x)=\Big (a- \frac{\gamma^2}{4} \Big) \cot \Big( \frac{x}{2} \Big) - \Big(b-a-\frac{\gamma^2}{4} \Big) \tan \Big( \frac{x}{2} \Big), \qquad x \in (0 , \pi) . $$
Since
$$ f'(x)= -\frac{1}{2}\Big (a- \frac{\gamma^2}{4} \Big) \left(1+\cot^2 \Big( \frac{x}{2} \Big) \right)  - \frac{1}{2}\Big(b-a-\frac{\gamma^2}{4} \Big) \left(1+\tan^2 \Big( \frac{x}{2} \Big) \right)$$
the mean value theorem implies that Assumption  \ref{as:A} is satisfied with $K=0$.
Now note that
\begin{align*}
f''(x) &= \frac{1}{2}\Big (a- \frac{\gamma^2}{4} \Big) \cot\Big( \frac{x}{2} \Big)  \left(1+\cot^2 \Big( \frac{x}{2} \Big) \right) \\ & \qquad  - \frac{1}{2}\Big(b-a-\frac{\gamma^2}{4} \Big) \tan\Big( \frac{x}{2} \Big)\left(1+\tan^2 \Big( \frac{x}{2} \Big) \right)
\end{align*}
and 
$$  (f'f)(x(t))+ \frac{\gamma^2}{2}f''(x(t))=  (f'f)(2\arcsin(y(t)))+ \frac{\gamma^2}{2}f''(2\arcsin(y(t))).         $$
Since
$$  \cot( \arcsin(y(t))) = \sqrt{ \frac{1-y(t)}{y(t)}},\qquad    \tan( \arcsin(y(t)))= \sqrt{ \frac{y(t)}{1-y(t)}}$$
and $y(t) \in (0,1)$ we obtain
$$  \left| f'(x(t))\right| \leq c \left( 1+ |y(t)|^{-1}  + |1-y(t)|^{-1}    \right)   $$
and
$$  \left| (f'f)(x(t))+ \frac{\gamma^2}{2}f''(x(t)) \right| \leq c \left( 1+ |y(t)|^{-3/2}  + |1-y(t)|^{-3/2} \right)$$
for some constant $c>0$, depending only on $a,b, \gamma >0$.
 Using Theorem 4.1 in \cite{hurd2008}
and establishing   uniform convergence of the given series expressions in $t$
 using asymptotic bounds on the Jacobi polynomials we have that
$$  \sup_{t \in [0,T]} \mathbb{E} |y(t)|^{q_1} < \infty $$
if 
$  q_1 >  -  \frac{2a}{\gamma^2} $
and 
$$  \sup_{t \in [0,T]} \mathbb{E} |1-y(t)|^{q_2} < \infty$$
if $ q_2 >  - \frac{2(b- a)}{\gamma^2}  $.
So Assumption \ref{as:MA} is satisfied if
$$  \frac{4}{3 \gamma^2}\min \{ a,b-a   \} >p. $$
\smallskip

Now Theorem  \ref{th:mainTH} gives
\begin{equation*}
\E\left[  \sup_{k=0, \ldots ,\lceil T / \Delta t \rceil}  \lev x(t_{k}) -  X_{k}\rev^{p} \right] \le C_{p}  \cdot \D t^{p} \qquad \textrm{for} \qquad  p < \frac{4}{3 \gamma^2}\min \{ a,b-a \}
\end{equation*}
 and transforming back yields
\begin{equation*}
\E\left[  \sup_{k=0, \ldots ,\lceil T / \Delta t \rceil}  \lev y(t_{k}) -  \sin^2\Big( \frac{X_{k}}{2}\Big)\rev^{p} \right] \le C_{p}  \cdot \D t^{p} .
\end{equation*}
under the same assumption on the parameters
since the $\sin$-function is bounded.

\smallskip
\begin{prop}
Let $2 \leq p <  \frac{4}{3 \gamma^2}\min \{ a,b-a   \}$. Then the LBE approximation of the Wright-Fisher  process is $p$-strongly convergent
with order one.
\end{prop}
\smallskip

\subsection{Ait-Sahalia model}

Higham et al.~analysed in \cite{szpruch-strongly} a backward Euler scheme for the Ait-Sahalia interest rate model 
\begin{equation}   \label{ait}
dy(t)=(\alpha _{-1}y(t)^{-1}-\alpha _{0}+\alpha
_{1}y(t)-\alpha_{2}y(t)^{r})dt+ \sigma y(t)^{\rho }dw(t),
\end{equation}
where $\alpha _{-1}, \alpha _{0}, \alpha, \alpha_{1}, \alpha_{2},  \sigma $ are positive
constants and $\rho,r > 1$. In \cite{szpruch-strongly} it was established that 
$$ \mathbb{P}(y(t) \in (0,\8), \,\, t \geq 0 ) =1.$$ 
Under the assumption $r+1>2\rho$ Higham et al.~ proved uniform  $L^{p}$-convergence for any $p \geq 2$ 
the backward Euler method  directly
applied to \eqref{ait}. However, 
their results did not reveal a rate of  convergence. Here, using the Lamperti transformation  approach we construct a scheme that strongly converges with  rate one. We focus 
on the critical case with $r=2$ and $\rho=1.5$ which was not covered in \cite{szpruch-strongly}. 
Using $F(y)= y^{-1/2}$
we obtain 
$$ dx(t) = f(x)dt 
   - \frac{1}{2}\sigma dw(t).  $$
with
$$ f(x) = \left(\frac{1}{2}\a_{2} + \frac{3}{8} \s^{2} \right) x(t)^{-1} - \frac{1}{2}\a_{1}x(t) + \frac{1}{2} \a_{0} x(t)^3 - \frac{1}{2}\a_{-1} x(t)^{5}, \quad  x>0.  $$
We have for 
$$ f'(x) =  -\left(\frac{1}{2}\a_{2} + \frac{3}{8} \s^{2} \right) x(t)^{-2} - \frac{1}{2}\a_{1} + \frac{3}{2} \a_{0} x(t)^2 
- \frac{5}{2}\a_{-1} x(t)^{4},  $$
that
$$ \lim_{x \rightarrow 0} f'(x)= - \infty, \qquad \lim_{x \rightarrow \infty } f'(x) = - \8$$
and hence there exists  a $c>0$ such that 
$$ \sup_{x >0}f'(x) < c.$$
Now the mean value theorem implies
\[
 (x-y) (f(x) - f(y)) \le c \lev x- y\rev^{2}, \qquad x,y>0,
\] i.e. the drift coefficient is one-sided Lipschitz.
We have moreover that
$$ f''(x) =  \left(\a_{2} + \frac{3}{4} \s^{2} \right) x(t)^{-3}  + 3 \a_{0} x(t) - 10\,\a_{-1} x(t)^{3},  $$
and 
$$  \left| (f'f)(x(t))+ \frac{\s^2}{8}f''(x(t)) \right| \leq c \left( 1+ |x(t)|^{-3}  + |x(t)|^{20} \right).$$
Straightforward computations also give that
$$ f(x) >  g(x), \qquad x \in (0, 2x(0)] $$
where
$$ g(x):=\left(\frac{1}{2}\a_{2} + \frac{3}{8} \s^{2} \right) x^{-1} -  \frac{1}{2} \alpha_2 \beta x $$
with
$$ \beta= \frac{\alpha_1}{\alpha_2} +  16\frac{\alpha_{-1}}{\alpha_2} x(0)^4. $$
 Now a comparison result for SDEs, see e.g. Proposition V.2.18  and Exercise V.2.19  in \cite{karatzas1991brownian},
yields that almost surely
$$ x(t)  \geq x^{(1)}(t) , \qquad t \in [0, \tau^{(1)}),$$
where
$$  d x^{(1)}_n(t)=g(x^{(1)}(t)) dt  - \frac{1}{2}\sigma dw(t), \quad t \geq 0,\qquad x^{(1)}(0)=x(0)$$
and
$$ \tau^{(1)}=  \inf \{ t \in [0,T]: x(t) > 2x(0) \} .$$
Let us define a sequence of stopping times
$$ \tau^{(2i)}=  \inf \{ t \in [\tau^{(2i-1)},T]: x(t) \le 2x(0) \}, \quad i=1,2,\ldots, $$
$$ \tau^{(2i+1)}=  \inf \{ t \in [\tau^{(2i)},T]: x(t) > 2x(0) \}, \quad i=1,2,\ldots $$
and an associated sequence of SDEs
$$  d x^{(i)}(t)=g(x^{(i)}(t)) dt  - \frac{1}{2}\sigma dw(t), \quad t \geq \tau^{(2i)} ,
\qquad x^{(i)}(\tau^{(2i)})=x(\tau^{(2i)}).$$
Using the comparison result for SDEs again, we have that almost surely
$$ x(t)  \geq x^{(i)}(t) , \qquad t \in [\tau^{(2i)}, \tau^{(2i+1)}).$$
Therefore, using the CIR process as a lower bound, Assumption \ref{as:MA} and consequently Theorem \ref{th:mainTH} holds for $ p< \frac{4}{3} \left( 1 + \frac{\alpha_2}{\sigma^2}\right) $.
Proceeding as for the $3/2$-model we have:

\begin{prop}
Let $1 \leq p < \frac{1}{3} + \frac{1}{3}\frac{\a_{2}}{\s^2}$. The LBE approximation of the Ait-Sahalia process with $r=2$ and $\rho=1.5$ is $p$-strongly convergent
with order one.
\end{prop}

\smallskip

In the case $r+1>2\rho$ we know from \cite{szpruch-strongly} that 
$$ \sup_{t\in[0,T]} \E \lev x(t) \rev^{-p} < \infty \qquad \text{for all} \qquad   p\geq 1.       $$
Moreover, the drift coefficient of the transformed SDE  behaves at zero like the one of a transformed CEV process.
A by now standard analysis gives:

\smallskip

\begin{prop}
Let $p \geq 1$. The LBE approximation of the Ait-Sahalia process with $r+1>2\rho$ is $p$-strongly convergent
with order one.
\end{prop}

\smallskip

\smallskip

\section{A Milstein-type scheme for CIR}

In this section we establish a connection between the Lamperti-backward Euler and a drift-implicit Milstein scheme for the CIR process.
We will show that the order of convergence of the LBE  carries over to a drift-implicit Milstein scheme, which
has been proposed in \cite{kahl2008structure} and \cite{higham2012convergence}. While strong convergence was shown in
\cite{compfin2012}, sharp convergence rates have not been established so far. 

Recall that BEM for the transformed CIR process reads as
\begin{align*} 
X_{k+1}  =  X_k + f(X_{k+1})\Delta t +\frac{1}{2}\s \Delta w_{k+1}, \qquad k=0,1, \ldots
\end{align*}
with
$$f(x)=\frac{1}{2}\k \left( \Big(\o - \frac{\s^{2}}{4\k} \Big) x^{-1} - x  \right), \qquad x>0.$$
Squaring yields the LBE, i.e.
\begin{align*} 
X_{k+1}^{2}  =   X_{k}^{2}  &+    \k( \o - X_{k+1}^{2}) \D t  +     \s X_{k}  \D w_{k+1} \\ & +
\frac{\s^2}{4} \big( (\D w_{k+1})^2 - \D t \big)  - \left( f(X_{k+1}) \right)^2\D t^2 .
\end{align*} 
On the other hand the drift-implicit Milstein scheme for CIR  is given by 
\begin{align} \label{eq:CIRMILST} 
Z_{k+1} =  Z_{k}  +    \k( \o - Z_{k+1}) \D t  +     \s \sqrt{Z_{k}}  \D w_{k+1} +
\frac{\s^2}{4} \big( (\D w_{k+1})^2 - \D t \big),
\end{align}  hence both schemes coincide up to a term of order $\D t^2$.
The numerical flows of the LBE  and Milstein scheme are given by
\begin{align*}
 \phi_{E}(x,k) &= \frac{1}{1+ \kappa \Delta t} \left(x  +    \k \o  \D t  +     \s \sqrt{x}  \D w_{k+1} +
\frac{\s^2}{4} \big( (\D w_{k+1})^2 - \D t \big) \right) \\ &  \qquad - \frac{1}{1+ \kappa \Delta t} \left( f(\phi_{E}(x,k) \right)^2\D t^2 
\end{align*}
and
\[
 \phi_{M}(x,k) = \frac{1}{1+ \kappa \Delta t} \left(x  +    \k \o  \D t  +     \s \sqrt{x}  \D w_{k+1} +
\frac{\s^2}{4} \big( (\D w_{k+1})^2 - \D t \big) \right) .
\]
It is clear then that
\[
 \phi_{M}(x,k) \ge  \phi_{E}(x,k)
\] for all $x>0$, $k =0,1, \ldots $.
From \cite{alfonsi2005discretization} we know on the other hand 
\[
\phi_{E}(x,k) \ge \phi_{E}(y,k) \quad \text{for} \quad  x\ge y. 
\]
Hence we conclude 
\begin{equation} \label{eq:pos}
 Z_{k} \ge X^{2}_{k}, \quad k=0,1,\ldots,
\end{equation}
so the drift-implicit Milstein scheme dominates the Lamperti-Euler method and thus  preserves positivity.

\smallskip

To  establish  the order of $L^1$-convergence for the drift-implicit Milstein scheme it is enough to 
control the difference between  the Lamperti-Euler method and  \eqref{eq:CIRMILST}.   

\smallskip
\begin{lemma}
 Let $\frac{ \kappa \theta}{\sigma^2} > 3/2$. Then there exists a 
constant $c>0$ such that
\begin{equation}  \label{eq:bem_milst}
 \sup_{k=0, \ldots ,\lceil T / \Delta t \rceil} \E    \lev Z_{k} -  X_{k}^2\rev  \le c \D t .
\end{equation}
\end{lemma}
\begin{proof}
 Let  $e_k = Z_k - X_k^2$ and note that $e_k \geq 0$ by \eqref{eq:pos}. We have
$$ e_{k+1} =e_k   - \kappa e_{k+1} \Delta t  + \s (\sqrt{Z_{k}} -  X_k ) \D w_{k+1} + (f(X_{k+1}))^{2} \Delta t^2. $$ 
Exploiting  the independence of
$X_{k}, Z_{k}$ and $\D w_{k+1}$ it follows
$$ \E e_{k+1} = \frac{1}{1+ \kappa \Delta t} \left( \E e_k  + \E (f(X_{k+1}))^{2} \Delta t^2 \right)$$
and consequently
$$   \sup_{k=0, \ldots ,\lceil T / \Delta t \rceil}  \E e_{k} \leq \D t \sum_{k=0}^{\lceil T / \Delta t \rceil -1}
\E (f(X_{k+1}))^2 \D t.$$
Due to our assumptions Lemma \ref{bound_inv} gives that
$$  \sup_{k=1, \ldots, \lceil T / \D t \rceil} \E X_k^{-2} \leq c,     $$
which together with 
$$ |f(x)| \leq c \cdot \left(1 + |x|+ |x|^{-1} \right)$$
and Lemma \ref{lem_sup_moments} shows the  assertion.

\end{proof}

\medskip

Using this result we have:
\begin{prop} \label{est_dimp}
(i) Let $ \frac{\kappa \theta}{\sigma^2} > 3/2$. Then,  there exists a constant $C >0$   such that
\begin{align}  \sup_{k=0, \ldots, \lceil T / \D t \rceil}\E   |y(k \D t) - Z_k | \leq C \cdot \D t. \label{dimp_est_L1} \end{align}
(ii)  Let $ \frac{\kappa \theta}{\sigma^2} > 3/2$. Then,  there exists a constant $C >0$   such that
\begin{align} \E   \sup_{k=0, \ldots, \lceil T / \D t \rceil}|y(k \D t) - Z_k |^2 \leq C \cdot \D t. \label{dimp_est_L2}\end{align}
\end{prop}
\begin{proof}

(i) This follows from the triangle inequality and Proposition \ref{subsec_cir}. 

(ii) Using \eqref{dimp_est_L1} the second assertion  can be shown along the lines of the proof of Proposition 5.3 in \cite{compfin2012},
where strong convergence of the drift-implicit Milstein scheme (without a convergence rate) was shown. 
Proceeding as in the proof of  Proposition 5.3 in \cite{compfin2012}
we have
\begin{align*} 
\mathbb{E} \sup_{k=0, \ldots,  \lceil T / \D t \rceil } |y(k \D t) -Z_k|^2   & \leq c  \sum_{\ell=0}^{ \lceil T / \D t \rceil-1}  (1+ \kappa \Delta t)^{2\ell} \mathbb{E} \left| y(\ell \D t)- Z_{\ell} \right| \Delta t  \\ & \qquad \quad \nonumber +   c  \, \mathbb{E} \sup_{k=1, \ldots,  \lceil T / \D t \rceil} \left| \sum_{\ell=0}^{k-1}  \frac{1}{(1+ \kappa \Delta t)^{k-\ell}} \rho_{\ell+1} \right|^2 
\end{align*}
with
$$ \rho_{k+1} = -\kappa \int_{k \D t}^{(k+1) \D t} (y(s)-y((k+1) \D t))\,ds + \sigma \int_{k \D t}^{(k+1) \D t} (\sqrt{y(s)} - \sqrt{y(k \D t)}  )\,dw(s).$$
So, \eqref{dimp_est_L1}  gives
\begin{align*} 
\mathbb{E} \sup_{k=0, \ldots,  \lceil T / \D t \rceil} |y(k \D t) -Z_k|^2   & \leq c \D t +   c  \, \mathbb{E} \sup_{k=1, \ldots,  \lceil T / \D t \rceil} \left| \sum_{\ell=0}^{k-1}  \frac{1}{(1+ \kappa \Delta t)^{k-\ell}} \rho_{\ell+1} \right|^2 .
\end{align*}
 For the second term straightforward computations yield that
\begin{align*}
 \mathbb{E} \sup_{k=1, \ldots,  \lceil T / \D t \rceil} \left| \sum_{\ell=0}^{k-1} \frac{1}{(1+ \kappa \Delta t)^{k-\ell}} \rho_{\ell+1} \right|^2 \leq c \Delta t, \end{align*}
which completes  the proof of the proposition. 
\end{proof}

\medskip

Using the (suboptimal) second estimate of the above Proposition and   Lemma 3.5  of \cite{dereich2012euler} 
we also obtain a sharp error estimate for the piecewise linear interpolation of the drift-implicit Milstein scheme, i.e.
\begin{align*} 
\overline{Z}_t=\frac{t_{k+1}-t}{\Delta}Z_{k}+ \frac{t -t_k}{\Delta}Z_{k+1}, \qquad t \in [t_k,t_{k+1}], \end{align*}
in a combined $L^2$-$\|\cdot\|_{\infty}$ norm.
 
\smallskip
\begin{prop} 
Let $ \frac{\kappa \theta}{\sigma^2} > 3/2$. Then,  there exists a constant $C >0$   such that
\begin{align*}\left( \E \max_{t \in [0,T]} |y(t) - \overline{Z}_t |^2  \right)^{1/2} \leq C \cdot \sqrt{\log(\Delta t)|} \cdot \sqrt{\Delta t}, \end{align*}
for all  $\Delta \in (0, 1/2]$.
\end{prop}

\medskip

 The above relation between a drift-implicit Milstein scheme applied to the original SDE
$$ dy(t)=a(y(t))dt + b(y(t)) d w(t) $$
and the BEM applied to the transformed SDE
\[
dx(t) = f(x(t)) dt +  \lambda dw(t)
\]
with
$$ f(x)= \lambda \left( \frac{a(F^{-1}(x))}{b(F^{-1}(x))} -\frac{1}{2}b'(F^{-1}(x)) \right) $$
and
\begin{equation*} 
F(x) = \lambda \int^{x} \frac{1}{b(y)}dy
\end{equation*}
is in fact a particular
case of a more general relation. Expanding  LBE  yields
\begin{align*}
 F^{-1}(X_{k+1}) =F^{-1}(X_{k}) & + \frac{1}{\lambda} b(F^{-1}(X_k))(X_{k+1}-X_k)  \\& + \frac{1}{2 \lambda^2}b'(F^{-1}(X_k)) b (F^{-1}(X_k)) (X_{k+1}-X_{k})^2 + \ldots,
\end{align*}
since
$$ \frac{d}{dx}F^{-1}(x)=  \frac{1}{\lambda} b(F^{-1}(x)), \qquad  \frac{d^2}{(dx)^2}F^{-1}(x)=  \frac{1}{\lambda^2} b'(F^{-1}(x)) b (F^{-1}(x)).
$$ Setting $Y_k=F^{-1}(X_{k})$ and using \eqref{BEM-lamp} we have
\begin{align*}
Y_{k+1} =Y_{k} &  +  \left( a(Y_{k+1}) - \frac{1}{2}b'b(Y_{k+1}) \right) \D t + b(Y_k) \D w_{k+1}   + \frac{1}{2}b'b(Y_k) |\D w_{k+1}|^2 \\ & + R_k
\end{align*}
with 
\begin{align*}
 R_k =  (b(Y_k) - b(Y_{k+1})) f(X_{k+1})\Delta t &+ \frac{1}{\lambda}
 (b'b)(Y_k) \Delta w_{k+1} f(X_{k+1}) \Delta t \\ & +
 \frac{1}{2\lambda^2}(b'b)(Y_k) f^2(X_{k+1}) \Delta t^2 + ...
\end{align*}

So dropping $R_{k}$ and the other higher order terms, we end up with
\begin{align*}
Z_{k+1} =Z_{k} &  +  \left( a(Z_{k+1}) - \frac{1}{2}b'b(Z_{k+1}) \right) \D t + b(Z_k) \D w_{k+1}   + \frac{1}{2}b'b(Z_k) |\D w_{k+1}|^2  .
\end{align*}
In the case of $(\alpha, \beta)=(0, \infty)$ conditions for the well-definedness, stability and strong convergence of this scheme are given
in   \cite{higham2012convergence}. However, the convergence rate analysis for the CIR process, where we can exploit (among other things) that $F^{-1}(x)=x^2$ and also the domination property
\eqref{eq:pos}, 
seems not to carry over to the general case.

\bigskip

\section{Conclusion and Discussion}

In this paper we presented a Lamperti-Euler scheme  for scalar SDEs which take values  in
a domain $D=(l,r)$ and have non-Lipschitz drift or  diffusion coefficients. Our strategy is to first 
use the Lamperti transformation $x(t) = F(y(t))$  (provided that the diffusion coefficient of the original SDE 
is strictly positive on $D$) 
and then to approximate the transformed process $x(t)$, $0\le t\le T$, with the backward Euler scheme.
Transforming back with the inverse Lamperti transformation  gives an approximation scheme for the original SDE.
We also  pointed out a relation of this scheme to a drift-implicit Milstein scheme.

\smallskip
The advantages of this Lamperti-Euler method are   
\begin{itemize}
 \item that it preserves the domain of the original SDE  
\item  and an available framework which allows
to establish strong convergence order one for this scheme.
\end{itemize}

In particular, we use this framework to obtain  such strong convergence results for several
SDEs with non-Lipschitz coefficients from both mathematical finance and bio-mathematics.

\smallskip

Whether the implicitness of this scheme (which e.g. for the CEV process requires solving a non-linear equation) can be avoided
by using a tamed Euler scheme (as in \cite{arnulf_tamed} for the case $D=\mathbb{R}$) 
remains an open question. Open questions are also, whether  Assumption \ref{as:A} and \ref{as:MA} can be formulated
in terms of conditions on the original coefficients of the SDE, 
and whether the convergence rate for the Lamperti-Euler also carries over in general (and not only for the CIR process) to a drift-implicit Milstein scheme.
In particular, the last point leads to the general
 question 
that given a certain numerical approximation (for an SDE with non-Lipschitz coefficients) which   
 perturbations do not change its convergence properties and also its qualitative properties? We will pursue all these topics in our future research.

\bigskip
\bigskip
\bigskip

{\bf Acknowledgements.} The authors would like to thank Martin Altmayer for
valuable comments on an earlier version of the manuscript.

\bigskip

\bibliographystyle{plain}
\bibliography{upthesis}

\end{document}